\documentclass{llncs}
\usepackage{amsmath}
\usepackage{amsfonts}
\usepackage{amssymb}
\usepackage{booktabs} 
\usepackage{xspace} 
\usepackage{tikz} 
\usepackage{times}
\usepackage{multirow}
\usepackage{microtype}

\usepackage{flushend}

\newcommand\is{witness}

\newcommand\iw{$i$-witness}
\newcommand\cu{\ensuremath{C_\_}}

\newcommand\even{\mathsf{even}}
\newcommand\val{\mathsf{value}}
\newcommand\up{\mathsf{up}}
\newcommand\ru{\mathsf{ru}}
\newcommand\au{\mathsf{au}}

\newcommand\won{\mathsf{won}}
\newcommand\lift{\mathsf{lift}}
\newcommand\bfb{\mathbf{b}}

\begin{document}
\title{An Ordered Approach to Solving Parity Games in Quasi Polynomial Time and Quasi Linear Space}






\author{John Fearnley\inst{1}, Sanjay Jain\inst{2}, Sven Schewe\inst{1}, Frank
Stephan\inst{2},
Dominik Wojtczak\inst{1}}

\institute{University of Liverpool \and National University of Singapore}

\sloppy



\newcommand \Aa {\mathcal{P}}
\newcommand{\seq}[1]{\langle #1 \rangle}

\newcommand{\Plays}{\mathsf{Plays}}

\maketitle

\begin{abstract}
Parity games play an important role in model checking and synthesis.
In their 
paper, Calude et al.\ have recently shown that these games can be solved in quasi-polynomial time.
We show that their algorithm can be implemented efficiently:
we use their data structure as a progress measure, allowing for a backward implementation instead of a complete unravelling of the game.
To achieve this, a number of changes have to be made to their techniques, where the main one is to add power to the antagonistic player that allows for determining her rational move without changing the outcome of the game.
We provide a first implementation for a quasi-polynomial algorithm, test it on small examples, and provide a number of side results, including minor algorithmic improvements, a quasi bi-linear complexity in the number of states and edges for a fixed number of colours, and matching lower bounds for the algorithm of Calude et al.
\end{abstract}

\section{Introduction}

Parity games are two-player 
zero-sum games played on a finite graph. The
two players, named \emph{even} and \emph{odd}, move a token around the graph
until a cycle is formed. Each vertex is labelled with an integer \emph{colour},
and the winner is determined by the \emph{parity} of the largest colour that
appears on the cycle: player \emph{even} wins if it is an even colour, and player
\emph{odd} wins otherwise.

Parity games have been the focus of intense study
\cite{Emerson+Lei/86/Parity,Emerson+Jutla/91/Memoryless,McNaughton/93/Games,Browne-all/97/fixedpoint,Zielonka/98/Parity,Jurdzinski/00/ParityGames,Ludwig/95/random,Puri/95/simprove,Voge+Jurdzinski/00/simprove,BjorklundVorobyov/07/subexp,Obdrzalek/03/TreeWidth,Lange/05/ParitySAT,Berwanger+all/06/ParityDAG,Jurdzinski/06/subex,Schewe/08/improvement,Fearnley/10/snare,DBLP:conf/icalp/ScheweTV15,DBLP:conf/lics/ChatterjeeHL15,Schewe/17/parity,quasipolynomial,JL17}, in part due to their
practical applications. Solving parity games is the central and most expensive
step in many model
checking~\cite{Kozen/83/mu,Emerson+Lei/86/Parity,Emerson+all/93/mu,Wilke/01/Alternating,deAlfaro+Henziger+Majumdar/01/Control,Alur+Henziger+Kupferman/02/ATL},
satisfiability
checking~\cite{Wilke/01/Alternating,Kozen/83/mu,Vardi/98/2WayAutomata,Schewe+Finkbeiner/06/ATM},
and synthesis~\cite{Piterman/06/Parity,Schewe+Finkbeiner/06/Asynchronous,DBLP:conf/cav/HahnSTZ16}
algorithms.

Parity games have also attracted attention due to their unusual complexity
status. The problem of determining the winner of a parity game is known to lie
in UP~$\cap$~co-UP~\cite{Jurdzinski/98/UP}, so the problem is very unlikely to
be NP-complete. However, despite much effort, no polynomial time algorithm has
been devised for the problem. Determining the exact complexity of solving a
parity game is a major open problem.

Three main classes of algorithms have been developed for solving parity games in
practice. The \emph{recursive}
algorithm~\cite{McNaughton/93/Games,Zielonka/98/Parity}, which despite being one
of the oldest algorithms has been found to be quite competitive in
practice~\cite{FriedmannL09}. \emph{Strategy improvement} algorithms use a local
search technique~\cite{Voge+Jurdzinski/00/simprove}, similar to the simplex
method for linear programming and policy iteration algorithms for solving Markov
decision processes. \emph{Progress measure} algorithms define a measure that
captures the winner of the game, and then use value iteration techniques to find
it~\cite{Jurdzinski/00/ParityGames}. Each of these algorithms has inspired lines
of further research, all of which have contributed to our understanding of
parity games. Unfortunately, all of them are known to have exponential worst
case complexity.

Recently, 
Calude et al.\
\cite{quasipolynomial} have provided a \emph{quasi-polynomial} time algorithm for
solving parity games that runs in time $O(n^{\lceil \log (c) + 6 \rceil})$,
where $c$ denotes the number of priorities used in the game. Previously, the
best known algorithm for parity games was a deterministic sub-exponential
algorithm~\cite{Jurdzinski/06/subex}, which could solve parity games in
$n^{O(\sqrt{n})}$ time, so this new result represents a significant advance in
our understanding of parity games. 

Their approach is to provide a compact witness that can be used to decide
whether player \emph{even} wins a play. Traditionally, one must store the entire
history of a play, so that when the players construct a cycle, we can easily
find the largest priority on that cycle. The key observation of Calude et al.\ \cite{quasipolynomial}
is that a witness of poly-logarithmic size can be used instead. This allows them
to simulate a parity game on an alternating Turing machine that uses
poly-logarithmic space, which leads to a deterministic algorithm that uses
quasi-polynomial time and space.

This new result has already inspired follow-up work. Jurdzi{\'n}ski and
Lazi{\'c}~\cite{JL17} have developed an adaptation of the classical
small-progress measures algorithm~\cite{Jurdzinski/00/ParityGames} that runs in
quasi-polynomial time. Their approach is to provide a succinct encoding of a
small-progress measure, which is very different from the succinct
encoding developed by Calude et al.\ \cite{quasipolynomial}.
The key advantage of using progress measures
as a base for the algorithm is that it avoids the quasi-polynomial space
requirement of the algorithm of Calude et al., instead providing an algorithm
that runs in quasi-polynomial time and near linear space.

\paragraph{Our contribution.}

In this paper, we develop a progress-measure based algorithm for solving parity
games that uses the succinct witnesses of Calude et al.\ \cite{quasipolynomial}.
These witnesses were designed to be used in a \emph{forward} manner, which means that they are
updated as we move along a play of the game.
Our key contribution is to show that these witnesses can also be used in a \emph{backwards} manner, by
processing the play backwards from a certain point. This allows us to formulate
a value iteration algorithm that uses (backwards versions of) the witnesses of
Calude et al.\ \cite{quasipolynomial} directly.

The outcome of this is to provide a second algorithm for parity games that runs
in quasi-polynomial time and near linear space. We provide a comprehensive
complexity analysis of this algorithm, which is more detailed than the one given
Calude et al.\ \cite{quasipolynomial} for the original algorithm.
In particular, we show that our algorithm provides
\begin{enumerate}
 \item a quasi bi-linear running time for a fixed number of colours, $O(m n \log(n)^{c-1})$;
 \item a quasi bi-linear FPT bound, e.g.\ $O(m n \mathfrak{a}(n)^{\log \log n})$, where any other quasi-constant function can be used to replace the inverse Ackermann function $\mathfrak{a}$; and
 \item an improved upper bound for a high number of colours, $O(m \cdot h \cdot n^{c_{1.45}+\log_2(h)})$ 
\end{enumerate}
for parity games with $m$ edges, $n$ vertices, and $c$ colours, where $h = \lceil 1 + c/\log(n) \rceil$ and
the constant $c_{1.45} = \log_2 \mathsf e < 1.45$.
We also provide an argument that parity games with $O(\log n)$ colours can be solved in polynomial time.

The complexity bounds (1) of our algorithm only match the bounds for the
algorithm of Jurdzi{\'n}ski and Lazi{\'c} \cite{JL17}, while (2) and (3) are
new. Moreover, we believe that it is
interesting that the witnesses of Calude et al.\ \cite{quasipolynomial} can be used in this way.
The history of research into parity games has shown that ideas from the varying
algorithms for parity games can often spur on further research. Our result and
the work of Jurdzi{\'n}ski and Lazi{\'c} show that there are two very different
ways of succinctly encoding the information that is needed to decide the winner
in a parity game, and that both of them can be applied in value iteration
algorithms. Moreover, implementing our progress measure is easier, as standard representations of the colours can be used. We have implemented our algorithm, and we provide some experimental results in the last section.

Finally, we present a lower bound for our algorithm, and for the algorithm of
Calude et al.\ \cite{quasipolynomial}. We derive a family of examples upon which both of the algorithms
achieve their worst case---quasi-polynomial---running time.
These are simple single player games.









\section{Preliminaries}
$\mathbb N$ denotes the set of positive natural numbers
$\{1,2,3,\ldots\}$.
Parity games are turn-based zero-sum games played
between two players---even and odd, or maximiser and minimiser---over finite graphs.  
A parity game $\Aa$ is a tuple $(V_e, V_o, E, C, \phi)$, where
$(V = V_e \cup V_o, E)$ is a finite directed graph with the set of
vertices $V$ partitioned into a set $V_e$ of vertices controlled by player \emph{even}
and a set $V_o$ of vertices controlled by player \emph{odd}, 
$E \subseteq V \times V$ is the set of edges,  $C\subseteq \mathbb N$ is a set of colours, and
$\phi: V \to C$ is the colour mapping. 
We require that every vertex has at least one outgoing edge. 

A parity game $\Aa$ is played between the two players, \emph{even} and \emph{odd}, by moving a token along
the edges of the graph. 
A play of such a game starts by placing a token on some initial vertex 
$v_0 \in V$.
The player controlling this vertex then chooses a successor vertex
$v_1$ such that $(v_0, v_1) \in E$ and the token is moved to this successor vertex. 
In the next turn the player controlling the vertex $v_1$ chooses the successor
vertex $v_2$ with $(v_1, v_2) \in E$ and the token is moved accordingly. 
Both players move the token over the arena in this manner and thus form a play
of the game.
Formally, a play of a game $\Aa$ is an infinite sequence of vertices
$\seq{v_0, v_1, \ldots} \in V^\omega$ such that, for all $i \geq 0$, we have that
$(v_i, v_{i+1}) \in E$.  
We write $\Plays_\Aa(v)$ for the set of plays of the game $\Aa$ that start from a vertex
$v \in V$ and $\Plays_\Aa$ for the set of plays of the game. 
We omit the subscript when the arena is clear from the context. 
We extend the colour mapping $\phi: V \to C$ from vertices to plays by defining
the mapping $\phi: \Plays \to C^\omega$ as  
$\seq{v_0, v_1, \ldots} \mapsto \seq{\phi(v_0), \phi(v_1), \ldots}$.

A play $\seq{v_0, v_1, \ldots}$ is won by player \emph{even} if $\limsup_{i \rightarrow \infty}\phi(v_i)$ is even, by player \emph{odd} if $\limsup_{i \rightarrow \infty}\phi(v_i)$ is odd.

A strategy for player \emph{even} is a function $\sigma: V^*V_e \rightarrow V$ such that $\big(v,\sigma(\rho,v)\big)\in E$ for all $\rho \in V^*$ and $v \in V_e$.
A strategy $\sigma$ is called memoryless if $\sigma$ only depends on the last state ($\sigma(\rho,v) = \sigma(\rho',v)$ for all $\rho,\rho' \in V^*$ and $v \in V_e$).
A play $\seq{v_0, v_1, \ldots}$ is consistent with $\sigma$ if, for every initial sequence $\rho_n = v_0, v_1, \ldots,v_n$ of the play that ends in a state of player \emph{even} ($v_n \in V_e$), $\sigma(\rho_n)=v_{n+1}$ holds.

It is well known that the following conditions are equivalent:
Player \emph{even} wins the game starting at $v_0$ if she has a strategy $\sigma$ that satisfies that
\begin{enumerate}
 \item all plays $\seq{v_0, v_1, \ldots}$ consistent with $\sigma$ satisfy $\limsup_{i \rightarrow \infty}\phi(v_i)$ (i.e.\ the highest colour that occurs infinitely often in the play) is even;
 \item all plays $\seq{v_0, v_1, \ldots}$ consistent with $\sigma$ contain a winning loop $v_i, v_{i+1}, \ldots, v_{i+k}$, that satisfies $v_i = v_{i+k}$ and $\phi(v_i) \geq \phi(v_{i+j})$ for all natural numbers $j \leq k$;
 \item as (1), and $\sigma$ must be memoryless; or
 \item as (2), and $\sigma$ must be memoryless.
\end{enumerate}

We use different criteria in the technical part, choosing the one that is most convenient.

\section{QP Algorithms}
We discuss a variation of the algorithm of Calude et al.~\cite{quasipolynomial}.

In a nutshell, the algorithm keeps a data structure, the \is es, that encodes the existence of sequences of ``good'' events.
This intuitively qualifies \is es as a measure of progress in the construction of a winning cycle. 
This intuition does not fully hold, as winning cycles are not normally identified immediately, but it gives a good intuition of the guarantees the data structure provides.

In~\cite{quasipolynomial}, \is es are used to track information in an alternating machine.
As they are quite succinct (they have only logarithmically many entries in the number of vertices of the game, and each entry only requires logarithmic space in the number of colours), this entails the quasi-polynomial complexity.

We have made this data structure accessible for value iteration, using it in a similar way as classical progress measures.
This requires a---simple---argument that \is es can be used in a backward analysis of a run just as well as in a forward analysis.
This, in turn, requires a twist in the updating rule that allows for rational decisions.
For this, we equip the data structure with an order, and show that the same game is still won by the same player if the antagonist can increase the value in every step.

\paragraph{\bf $i$-Witnesses} Let $\rho = v_1, v_2, \dots, v_m$ be a prefix of a play of the
parity game. An \emph{\iw} is a sequence of (not necessarily consecutive) positions
of $\rho$ 
\begin{equation*}
p_1, p_2, p_3, \dots, p_{2^i}, 
\end{equation*}
of length exactly $2^i$, that satisfies the following properties:
\begin{itemize}
\item \textbf{Position:} Each $p_j$ specifies a position in the play $\rho$,
so each $p_j$ is an integer that satisfies $1 \le p_j \le m$.
\item \textbf{Order:} The positions are ordered. So we have $p_j < p_{j+1}$ for
all $j < 2^i$.

\item {\bf Evenness:} 
All positions other than the final one are even. Formally, for all $j < 2^i$ the
colour $\phi(v_{p_{j}})$ of the vertex in position $p_{j}$ is even.

 \item {\bf Inner domination:}  
The colour of every vertex between $p_j$ and $p_{j+1}$ is dominated by the colour
of $p_j$, or the colour of $p_{j+1}$. Formally, 
for all $j < 2^i$, the largest colour of any vertex in the subsequence
$v_{p_{j}},v_{(p_{j})+1},\ldots,v_{p_{(j+1)}}$ is less than or equal to 
  $\max\big\{\phi(v_{p_{j}}),\phi(v_{p_{j+1}})\big\}$.

 \item {\bf Outer domination:} 
The colour of $p_{2^i}$ is greater than or equal to the colour of every vertex
that appears after 
$p_{2^i}$ in $\rho$. Formally, for all $k$ in the range $p_{2^i} < k \le m$, we
have that $\phi(v_{k}) \le \phi(v_{p_{2^i}})$.

\end{itemize}

\paragraph{\bf Witnesses}

We define $\cu = C \cup \{ \_ \}$ to be the set of colours augmented with the $\_$
symbol.
A \emph{witness} is a sequence 
\begin{equation*}
b_k, b_{k-1}, \dots, b_1, b_0,
\end{equation*}
of length $k+1$---we will later see that $k = \lfloor \log_2(e)\rfloor$ is big enough, where $e$ is the number of vertices with an even colour---where each element $b_i \in \cu$, and
that satisfies the following properties.
\begin{itemize}
\item \textbf{Witnessing.} There exists a family of $i$-witnesses, one for each element $b_i$ with $b_i \ne \_$.
We refer to such an $i$-witness in the run $\rho$. We will refer to this witness as
\begin{equation*}
p_{i, 1}, \; p_{i, 2}, \; \dots, \; p_{i, 2^i}.
\end{equation*}
\item \textbf{Dominating colour.}
For each $b_j \ne \_$, we have that $b_j = \phi(v_{p_{i, 2^i}})$. In other words,
$b_j$ is the outer domination colour of the $i$-witness.
\item \textbf{Ordered sequences.} The $i$-witness associated with $b_i$ starts after
$j$-witness associated with $b_{j}$ whenever $i < j$. Formally, for all $i$
and $j$ with $i < j$, if $b_i
\ne \_$ and $b_{j} \ne \_$, then $p_{j, 2^j} < p_{i, 1}$.
\end{itemize}
It should be noted that the $i$-witnesses associated with each position $b_i$ are not
stored in the witness, but in order for a sequence to be a witness, the
corresponding $i$-witnesses must exist.

Observe that the dominating colour property combined with the ordered
sequences property imply that the colours in a witness are monotonically
increasing, since each colour $b_j$ (weakly) dominates all colours that appear
afterwards in~$\rho$.

\paragraph{\bf Forwards and backwards witnesses.} So far, we have described
\emph{forwards} witnesses, which were introduced in~\cite{quasipolynomial}. In
this paper, we introduce the concept of a  
\emph{backwards} witnesses, and an ordering over these witnesses, which will be used in our main result. For each play $\rho = v_1, v_2,
\dots, v_m$, we define the reverse play $\overleftarrow{\rho} = v_m, v_{m-1},
\dots, v_1$. A backwards witness is a witness for $\overleftarrow{\rho}$, or for an initial sequence of it.

\paragraph{\bf Order on witnesses.}
We first introduce an order $\succeq$ over the set $\cu$ that captures the
following requirements: even numbers are better than odd numbers, and all
numbers are better than $\_$. Among the even numbers, higher numbers are better
than smaller ones, while among the odd numbers, smaller numbers are better than
higher numbers. Formally, $b \succeq c$ if either $c = \_\,$; or if $c$ is odd and
$b$ is either odd and $b \leq c$ holds, or $b$ is even; or
$c$ is even and $b$ is even and $b \geq c$ holds.

Then, we define an order $\sqsupseteq$ over witnesses. This order compares two
witnesses lexicographically, starting from $b_k$ and working downwards, and for each individual position the entries are compared using $\succeq$.
We also define a special witness $\won$ which is $\sqsupseteq$ than any other witness.

\paragraph{\bf The value of a witness.}
An \emph{even chain} of length $m$ is a sequence of positions $p_1 < p_2 <p_3 <
\ldots < p_m$ (with $0\leq p_0$ and $p_m \leq n$) in $\rho$ that has the
following properties:
\begin{itemize}
\item for all $j \leq m$, we have that $\phi(v_{p_j})$ is even, and
\item for all $j < m$ the colours in the subsequence defined by $p_{j}$ and
$p_{j+1}$ are less than or equal to  $\phi(p_j)$ or $\phi(p_{j+1})$. More
formally, we have that all colours
$\phi(v_{p_j}),\phi(v_{(p_j)+1}),\ldots,\phi(v_{p_{(j+1)}})$ are less than or
equal to $\max\big\{\phi(v_{p_j}),\phi(v_{p_{j+1}})\big\}$.
\end{itemize}

For each witness $\mathbf b = b_k, b_{k-1}, \ldots, b_0$, we define the function
$\even(\mathbf b,i) = 1$ if $b_i \ne \_$ and $b_i$ is even. Then we define the
value of the witness $\mathbf b$ to be $\val(\mathbf b) = \sum_{i=0}^k 2^i \cdot
\even(\mathbf b,i)$. We can show that the value $\mathbf{b}$ corresponds to the
length of an even chain in $\rho$ that is witnessed by $\mathbf{b}$.

\begin{lemma}
If\ $\mathbf{b}$ is a (forward or backward) witness of $\rho$, then there 
is an even chain of length $\val(\mathbf b)$ in $\rho$.
\end{lemma}
\begin{proof}
Let $i$ be an index such that $\even(\mathbf{b}, i) = 1$. By definition, the
$i$-witness $p_{i, 1}, p_{i, 2}, \dots, p_{i, 2^i}$ is an even chain of length
$2^i$ in $\rho$. This holds irrespective of whether $\mathbf{b}$ is a forward
or backward witness. 

Then, given an index $j > i$ such that $\even(\mathbf{b}, j) =
1$, observe that the outer domination property ensures that $\phi(p_{i, 2^i})
\ge \phi(v_l)$ for all $l$ in the range $p_{i, 2^i} \le l \le p_{j, 1}$. So,
when we concatenate the $i$-witness with the $j$-witness we still obtain an even
chain. Thus, $\rho$ must contain an even chain of length $\val(\mathbf b)$.
\end{proof}

Let $e = |\{v \in V \; : \; \phi(v) \text{ is even }\}|$ be the number of
vertices with even colours in the game. Observe that, if we have an even chain
whose length is strictly greater than $e$, then $\rho$ must contain a cycle,
since there must be a vertex with even colour that has been visited twice.
Moreover, the largest priority on this cycle must be even, so this is a winning
cycle for player \emph{even}. Thus, for player \emph{even} to win the parity game, it is
sufficient for him to force a play that has a witness whose value is strictly
greater than $e$.

\begin{lemma}
\label{lem:correct}
If, from an initial state $v_0$, player \emph{even} can force the game to run
through a sequence $\rho$, such that $\rho$ has a (forwards or backwards)
witness $\mathbf b$ such that $\val(\mathbf b)$ is greater than the number of vertices
with even colour, then player \emph{even} wins the parity game starting at $v_0$.
\end{lemma}


\subsection{Updating forward witnesses}


We now show how forward witnesses can be constructed incrementally by processing
the play one vertex at a time. Throughout this subsection, we will suppose that
we have a play $\rho = v_0, v_1, \dots, v_{m}$, and a new vertex $v_{m+1}$ that
we would like to append to $\rho$ to create $\rho'$. We will use $d =
\phi(v_{m+1})$ to denote the colour of this new vertex. We will suppose that
$\mathbf b = b_k, b_{k-1}, \ldots, b_1,b_0$ is a witness for $\rho$, and we will
construct a witness $\mathbf c = c_k, c_{k-1}, \ldots, c_1,c_0$ for $\rho'$.

We present three lemmas that allow us to perform this task.

\begin{lemma}
\label{lem:up.overflow}
Suppose that there exists an index $j$ such that $b_i$ is even for all $i < j$,
and that $b_i \ge d$ or $b_i = \_$ for all $i > j$. If we set $c_i = b_i$ for
all $i > j$, $c_j = d$, and $c_i = \_$ for all $i < j$, then $\mathbf{c}$ is a
witness for $\rho'$.
\end{lemma}

\begin{proof}
For the indices $i > j$, observe that since $b_i \ge d$, the outer domination of
the corresponding $i$-witnesses continues to hold. For the indices $i < j$,
since we set $c_i = \_$ there are no conditions that need to be satisfied.

To complete the proof, we must argue that there is a $j$-witness that
corresponds to $c_j$. This witness is obtained by concatenating the
$i$-witnesses corresponding to the numbers $b_i$ for $i < j$, and then adding
the vertex $v_{m+1}$ as the final position. This produces a sequence of length
$1 + \sum_{i = 0}^{j - 1} 2^i = 2^j$ as required. Since all $b_i$ with $i < j$
were even, the evenness condition is satisfied. For inner domination, observe
that the outer domination of each $i$-witness ensures that the gaps between the
concatenated sequences are inner dominated, and the fact that $b_0$ 
dominates sequence $v_{p_{0,1}},\ldots,v_m$ ensures that the final subsequence is also dominated by $b_0$ or $d$.
Outer domination is
trivial, since $v_{m+1}$ is the last vertex in $\rho'$. So, we have constructed
a $j$-witness for $\rho'$, and we have shown that $\mathbf{c}$ is a witness for
$\rho'$.
\end{proof}

Note that, differently from Calude et al.~\cite{quasipolynomial}, we also allow this operation
to be performed in the case where $d$ is odd.


\begin{lemma}
\label{lem:up.local}
Suppose that there exists an index $j$ such that $b_j \neq \_$, $d > b_j$, and, for all $i>j$, either $b_i = \_$ or  $b_i \geq d$ hold.
Then setting $c_i = b_i$ for all $i > j$, setting $c_j = d$, and setting $c_i = \_$ for all $i < j$ yields a witness for $\rho'$.
\end{lemma}

\begin{proof}
For all $i > j$, we set $c_i = b_i$. Observe that this is valid, since $b_i \geq d$, and so the outer domination property continues to hold for the
$i$-witness associated with $b_i$. For all $i < j$, we set $c_i = \_$, and this
is trivially valid, since this imposes no requirements upon $\rho'$.

To complete the proof, we must argue that setting $c_j = d$ is valid. Observe
that in $\rho$, the $j$-witness associated with $b_j$ ends at a certain position
$p = p_{j, 2^j}$. We can create a new $j$-witness for $\rho'$ by instead setting
$p_{j, 2^j} = m+1$, that is, we change the last position of the $j$-witness to
point to the newly added vertex. Note that inner domination continues to hold,
since $d > b_j = \phi(v_p)$ and since $v_p$ outer dominated $\rho$. All other
properties trivially hold, and so $\mathbf c$ is a witness for $\rho'$.
\end{proof}

\begin{lemma}
\label{lem:up.stale}
Suppose that for all $j \leq k$ either $b_j = \_$ or $b_j \geq d$. If we
set $c_i = b_i$ for all $i \leq k$, then $\mathbf c$ is a witness for $\rho'$.
\end{lemma}

\begin{proof}
Since $d \leq b_j$ for all $j$, the outer domination of every $i$-witness
implied by $\mathbf b$ is not changed. Moreover, no other property of a witness
is changed by the inclusion of $v_{m+1}$, so by setting $\mathbf c = \mathbf b$
we obtain a witness for $\rho'$.
\end{proof}

When we want to update a witness upon
scanning another state $v_{m+1}$, we find the largest witness that
(according to $\sqsubseteq$) can be obtained by applying Lemmas \ref{lem:up.overflow} through \ref{lem:up.stale}.
The largest such witness is quite easy to find: first, there are at most $3k$ to check, but the rule is simply to update the leftmost position in a witness that can be updated.

For a given witness $\mathbf b$ and a vertex $v_{m+1}$, we denote with
\begin{itemize}
 \item $\ru(\mathbf b,v_{m+1})$ the raw update of the witness to $\mathbf c$, as
obtained by the update rules described above.
 \item $\up(\mathbf b,v_{m+1})$ is either $\ru(\mathbf b,v_{m+1})$ if
$\val\big(\ru(\mathbf b,v_{m+1})\big) \leq e$ (where $e$ is the number of
vertices with even colour), or $\up(\mathbf b,v_{m+1}) = \won$ otherwise.
\end{itemize}

\section{Basic Update Game}
With these update rules, we define a forward and a backward basic update game.
The game is played between player \emph{even} and player \emph{odd}.
In this game, player \emph{even} and player \emph{odd} produce a play of the game as usual: if the pebble is on a position of player \emph{even}, then  player \emph{even} selects a successor, and if the pebble is on a position of player \emph{odd}, then  player \emph{odd} selects a successor.

Player \emph{even} can stop any time she likes and evaluate the game using $\mathbf b_0 = \_,\ldots,\_$ as a starting point and the update rule $\mathbf b_{i+1} = \up(\mathbf b_i,v_i)$.
For a forward game, she would process the partial play $\rho^+ = v_0, v_1, v_2, \ldots, v_n$ from left to right, and for the backwards game she would process the partial play $\rho^- = v_n, v_{n-1}, \ldots, v_0$.
In both cases, she has won if $\mathbf b_{n+1} = \won$.

\begin{theorem}
\label{theo:basic}
If player \emph{even} has a strategy to win the (forward or backward) basic update game, then she has a strategy to win the parity game.
\end{theorem}
\begin{proof}
By definition, we can only have $\mathbf b_{n+1} = \won$ if at some point we
created a witness whose value was more than the total number of even colours in
the game. As we have argued, such a witness implies that a cycle has been
created, and that the largest priority on the cycle is even. Since player
\emph{even} can achieve this no matter what player \emph{odd} does, this implies
that player \emph{even} has a winning strategy for the parity game.
\end{proof}


\section{The Data-structure for the Progress Measure}

Recall that there are two obstacles in implementing the algorithm of Calude et
al.~\cite{quasipolynomial} as a value iteration algorithm. The first (and minor) obstacle is
that it uses forward witnesses, while value iteration naturally uses backward
witnesses. We have already addressed this point by introducing the same measure
for a backward analysis.

The second obstacle is the lack of an order over witnesses that is compatible
with value iteration. While we have introduced an order in the previous
sections, this order is not a natural order. In particular, it is not preserved
under update, nor does it agree with the order over values. As a simple example
consider the following two sequences:
\begin{itemize}
 \item $\mathbf b = \_,4,2$, and
 \item $\mathbf c = 9,8,\_$.
\end{itemize}
While $\val(\mathbf b) = 3 > \val(\mathbf c) = 2$, $\mathbf c \sqsupset \mathbf
b$. In particular, $c_2 \succ b_2$ and $c_1 \succ b_1$ hold. Yet, when using the
update rules when traversing a state with colour $6$, $\mathbf b$ is updated to
$\mathbf b' = 6,\_,\_,$, while $\mathbf c$ is updated to  $\mathbf c' = 9,8,6$.
While $\mathbf c \sqsupset \mathbf b$ held prior to the update, $\mathbf b' \sqsupset \mathbf c'$ holds after the update.
Value iteration, however, needs a natural order that will allow us to choose
the successor with the higher value.

We overcome this problem by allowing the antagonist in our game, player
\emph{odd}, an extra move: prior to executing the update rule for a value
$\mathbf b$, player \emph{odd} may increase the witness $\mathbf b$ in the
$\sqsubseteq$ ordering. The corresponding 
\emph{antagonistic update} is defined as follows.
$$\au(\mathbf b,v) = {\min}_\sqsubseteq\big\{\up(\mathbf c,v) \mid \mathbf c \sqsupseteq \mathbf b \big\}$$

Obtaining $\au(\mathbf b,v)$ is quite simple: 
$\au(\mathbf b, v)$ is either
$\up(\mathbf b, v)$ or it is $\up(\mathbf d, v)$ where $\mathbf d$ is
$${\min}_\sqsubseteq \{ \mathbf {d} = d_k, d_{k-1}, \dots, d_0 \; ; \;
\mathbf{d} \sqsupseteq \mathbf{b} \text{ and } d_0 = \_\}.$$

Observe that if $\mathbf b \sqsubseteq \mathbf b'$ then $\au(\mathbf b,v)
\sqsubseteq \au(\mathbf b',v)$, because the minimum used in $\au(\mathbf b',v)$ ranges over a smaller set.

\section{Antagonistic Update Game}
The antagonistic update game is played like the basic update game, but uses the antagonistic update rule.
I.e.\ player \emph{even} and \emph{odd} play out a play of the game as usual: if the pebble is on a position of player \emph{even}, then  player \emph{even} selects a successor, and if the pebble is on a position of player \emph{odd}, then  player \emph{odd} selects a successor.

Player \emph{even} can stop any time she likes and evaluate the game using $\mathbf b_0 = \_,\ldots,\_$ as a starting point and the update rule $\mathbf b_{i+1} = \au(\mathbf b_i,v_i)$.
For a forward game, she would process the partial play $\rho^+ = v_0, v_1, v_2, \ldots, v_n$ from left to right, and for the backwards game she would process the partial play $\rho^- = v_n, v_{n-1}, \ldots, v_0$.
In both cases, she has won if $\mathbf b_{n+1} = \won$.

\begin{theorem}
\label{theo:antogonistic}
If player \emph{even} has a strategy to win the (forward or backward) antagonistic update game, then she has a strategy to win the parity game.
\end{theorem}

\begin{proof}
We first look at the evaluation of a play  $\rho^+ = v_0, v_1, v_2, \ldots, v_n$ or $\rho^- = v_n, v_{n-1}, \ldots, v_0$ in a forward or backwards game, respectively.
In an antagonistic game, this will lead to a sequence $\mathbf a_0, \mathbf a_1, \ldots, \mathbf a_{n+1}$, while it leads to a sequence $\mathbf b_0, \mathbf b_1, \ldots, \mathbf b_{n+1}$ when using the basic update rule.
We show by induction that $\mathbf b_i \sqsupseteq \mathbf a_i$ holds.

For an induction basis, $\mathbf b_0 = \mathbf a_0 = \_,\ldots,\_$.

For the induction step, if $\mathbf b_i \sqsupseteq \mathbf a_i$, then
\begin{align*}
\mathbf a_{i+1} = \au(\mathbf a_i,v_i) &= \min_\sqsubseteq\big\{\up(\mathbf
c,v_i) \mid \mathbf c \sqsupseteq \mathbf a_i \big\} \\
&\sqsubseteq \up(\mathbf a_i, v_i)  \\
&\sqsubseteq^{IH} \up(\mathbf b_i,v_i) = \mathbf b_{i+1}.
\end{align*}
Thus, when player \emph{even} wins the (forward or backward) antagonistic update game, then she wins the (forward or backward) basic update game using the same strategy.
\end{proof}

It remains to show that, if player \emph{even} has a strategy to win the parity
games, then she has a strategy to win the antagonistic update game. 
For this, we will use the fact that she can, in this case, make sure that the highest number that occurs infinitely often on a run is even.
We exploit this in two steps.
We first introduce a $\downarrow_x$ operator, for
every even number $x$, that removes all but possibly one entry with numbers smaller
than $x$, and adjust the one that possibly remains to $x-1$.
We then argue that, when there are no higher numbers than $x$, this value of the
witnesses obtained after this operator are non-decreasing w.r.t.\ $\sqsupseteq$,
and increase strictly with every occurrence of $x$.

Formally we define, for a witness $\mathbf b = b_k,b_{k-1},\ldots,b_0$ and an
even number $x$, the following.
\begin{itemize}
\item $\mathbf b\downarrow_x$ to be $\mathbf b$ if, for all $i\leq k$, $b_i = \_$ or $b_i \geq x$ holds.
\item 
Otherwise, let $i = \max\{s \leq k \mid b_{s} \neq \_$ and $b_{s} < x\}$. We define
$\mathbf b\downarrow_x = b_k',b_{k-1}',\ldots,b_0'$ with $b_j' = b_j$ for all
$j>i$, $b_i' = x-1$, and  $b_j' = \_$ for all $j<i$.
\end{itemize}

\begin{lemma}
\label{lem:up.game}
The $\downarrow_x$ operator provides the following guarantees:
\begin{enumerate}
 \item $\mathbf b \sqsupset \mathbf a \; \Rightarrow \; \mathbf b\downarrow_x \sqsupseteq \mathbf a\downarrow_x$
 \item $\phi(v) < x  \; \Rightarrow \; \up(\mathbf b,v)\downarrow_x \sqsupseteq \mathbf b\downarrow_x$
 \item $\phi(v) < x  \; \Rightarrow \; \au(\mathbf b,v)\downarrow_x \sqsupseteq \mathbf b\downarrow_x$
 \item $\phi(v) = x  \; \Rightarrow \; \up(\mathbf b,v)\downarrow_x \sqsupset \mathbf b\downarrow_x$
 \item $\phi(v) = x  \; \Rightarrow \; \au(\mathbf b,v)\downarrow_x \sqsupset \mathbf b\downarrow_x$
\end{enumerate}
\end{lemma}

\begin{proof}
For (1), let $i \leq k$ be the highest position with $b_i \neq a_i$, and thus with $b_i \succ a_i$ (as $\mathbf b \sqsupset \mathbf a$).
If $b_i \succeq x$ or $x+1  \succeq a_i$, the claim follows immediately (and we have $\mathbf b\downarrow_x \sqsupset \mathbf a\downarrow_x$).
For the case  $x \succ b_i \succ a_i \succ x+1$, this position would be replaced
by $x-1$ and all smaller positions by $\_$ by the $\downarrow_x$ operator (and we have $\mathbf b\downarrow_x = \mathbf a\downarrow_x$).

For (2), the highest position $i \leq k$ for which $\mathbf a =\up(\mathbf b,v)$ and $\mathbf b$ differ (if any) satisfies $a_i < x$ and $b_i \prec x$ (the latter holds because otherwise $v$ does not overwrite position $i$ by this update rule).
If $b_i  \prec x+1$, then we get $\up(\mathbf b,v)\downarrow_x \sqsupset \mathbf b\downarrow_x$; otherwise we get  $\up(\mathbf b,v)\downarrow_x = \mathbf b\downarrow_x$.

(3) follows from (1) and (2).

For (4),  $\mathbf a =\up(\mathbf b,v)$ and $\mathbf b$ differ in some highest position $i \leq k$, and for that position, $x = a_i \succ b_i$ holds. Thus, $\up(\mathbf b,v)\downarrow_x \sqsupset \mathbf b\downarrow_x$.

(5) follows with (1) and (4).
\end{proof}

This almost immediately implies the correctness.

\begin{theorem}
If player \emph{even} can win the parity game from a position $v$, then she can win the (forward and backward) antagonistic update game from $v$. 
\end{theorem}

\begin{proof}
Player \emph{even} can play such that the highest colour that occurs in a run infinitely many times is even.
She can thus in particular play to make sure that, at some point in the run, an
even colour $x$ has occurred more often that the size of the image of
$\downarrow_x$ after the last occurrence of a priority higher than $x$.
By Lemma \ref{lem:up.game}, evaluating the forward or backward antagonistic update game at this point will lead to a win of player \emph{even}.
\end{proof}

These results directly provide the correctness of all four games described.

\begin{corollary}
\label{cor:correctness}
Player \emph{even} can win the forward and backward antagonistic and basic update game from a position $v$ if, and only if, she can win the parity game from $v$.
\end{corollary}

\section{Value Iteration}
The antagonistic update game offers a direct connection to \emph{value iteration}.
For value iteration, we use a \emph{progress measure}, a function $\iota : V
\rightarrow \mathbb W$, where $\mathbb W$ denotes the set of possible backwards witnesses.
That is, a progress measure assigns a backwards witness to each vertex.

Let $\mathbf b_v = \max_\sqsubseteq\{\au(\iota(s),v) \mid (v,s) \in E\}$ for $v
\in V_e$ and $\mathbf b_v = \min_\sqsubseteq\{\au(\iota(s),v) \mid (v,s) \in
E\}$ for $v \in V_o$. We say that $\iota$ can be lifted at $v$ if $\iota(v)
\sqsubset \mathbf b_v$. When $\iota$ is liftable at $v$, we define by
$\lift(\iota,v)$ the function $\iota'$ with $\iota'(v) = \mathbf b_v$ and
$\iota'(v') = \iota(v')$ for all $v' \neq v$. We extend the lift operation to
every non-empty set $V' \subseteq V$ of liftable positions, where $\iota' =
\lift(\iota,V')$ updates all values $v \in V'$ concurrently.

A progress measure is called consistent if it cannot be lifted at any vertex $v \in V$.
The \emph{minimal consistent progress measure} $\iota_{\min}$ is the smallest (w.r.t.\ the partial order in the natural lattice defined by pointwise comparison) progress measure that satisfies
\begin{itemize}
 \item for all $v \in V_e$ that $\iota(v) \sqsupseteq \max_\sqsubseteq\{\au(\iota(s),v) \mid (v,s) \in E\}$, and 
 \item for all $v \in V_o$ that $\iota(v) \sqsupseteq \min_\sqsubseteq\{\au(\iota(s),v) \mid (v,s) \in E\}$.
\end{itemize}

As $\au(\mathbf b,v)$ is monotone in $\mathbf b$ by definition and the state space is finite, we get the following.

\begin{lemma}
\label{lem:well.defined}
The minimal consistent progress measure $\iota_{\min}$ is well defined. 
\end{lemma}

\begin{proof}
First, a consistent progress measure always exists:
the function that maps all states to $\won$ is a consistent progress measure.

Second if we have two consistent progress measures $\iota$ and $\iota'$, then the pointwise minimum $\iota'': v \mapsto \min_\sqsubseteq\{\iota(v),\iota'(v)\}$ is a consistent progress measure.
To see this, we assume w.l.o.g.\ that $\iota(v) \sqsubseteq \iota'(v)$.

For $v \in V_e$ we get
$\iota''(v) = \iota(v)  \sqsupseteq \max_\sqsubseteq\{\au(\iota(s),v) \mid (v,s) \in E\} \sqsupseteq \max_\sqsubseteq\{\au(\iota''(s),v) \mid (v,s) \in E\}$, using that $ \iota''(s) \sqsubseteq \iota(s)$ holds for all $s \in V$.

Likewise, we get for $v \in V_o$ that $\iota''(v) = \iota(v)  \sqsupseteq \min_\sqsubseteq\{\au(\iota(s),v) \mid (v,s) \in E\} \sqsupseteq \min_\sqsubseteq\{\au(\iota''(s),v) \mid (v,s) \in E\}$, using again that $ \iota''(s) \sqsubseteq \iota(s)$ holds for all $s \in V$.

As the state space is finite, we get the minimal consistent progress measure as a pointwise minimum of all consistent progress measures.
\end{proof}

Moreover, we can compute the minimal consistent progress measure
by starting with the initial progress measure $\iota_0$, which maps all vertices
to the minimal witness $\_,\ldots,\_$, and iteratively lifting.

\begin{lemma}
\label{lem:construct}
The minimal consistent progress measure $\iota_{\min}$ can be obtained by any sequence of lift operations on liftable positions, starting from $\iota_0$. 
\end{lemma}

\begin{proof}
We show that, for any sequence $\iota_0, \iota_1, \ldots, \iota_n$ of progress
measures constructed by a sequence of lift operations,
for all $v \in V$, and for all $i \leq n$, $\iota_i(v) \sqsubseteq \iota_{\min}(v)$ holds.

For the induction basis, $\iota_0(v)$ is the minimal element for all $v \in V$, such that  $\iota_0(v) \sqsubseteq \iota_{\min}(v)$ holds trivially.
For the induction step, let $V_i \subseteq V$ be a set of liftable position for $\iota_i$ and $\iota_{i+1} = \lift(\iota_i,V_i)$.
We now make the following case distinction.
\begin{itemize}
\item
For $v \in V_i\cap V_e$, we have $\iota_{i+1}(v) = \max_\sqsubseteq\{\au(\iota(s),v) \mid (v,s) \in E\} \sqsubseteq_{IH} \max_\sqsubseteq\{\au(\iota_{\min}(s),v) \mid (v,s) \in E\} \sqsubseteq \iota_{\min}(v)$.

\item 
For $v \in V_i\cap V_o$, we have $\iota_{i+1}(v) = \min_\sqsubseteq\{\au(\iota(s),v) \mid (v,s) \in E\} \sqsubseteq_{IH} \min_\sqsubseteq\{\au(\iota_{\min}(s),v) \mid (v,s) \in E\} \sqsubseteq \iota_{\min}(v)$.

\item For $v \notin V_i$, we have  $\iota_{i+1}(v) = \iota_i(v) \sqsubseteq_{IH} \iota_{\min}(v)$.
\end{itemize}
This closes the induction step.

While we have proven that the value of the progress measures cannot surpass the
value of $\iota_{\min}$ at any vertex, each liftable progress measure $\iota_i$
is succeeded by a progress measure  $\iota_{i+1}$, which is nowhere smaller, and strictly increasing for some vertices. Thus, this sequence terminates eventually by reaching a non-liftable progress measure.
But non-liftable progress measures are consistent.

Thus, we eventually reach a consistent progress measure $\iota_n$ which is pointwise no larger than $\iota_{\min}$;
i.e.\ we eventually reach $\iota_{\min}$.
\end{proof}

It is simple to get from establishing that $\iota_{\min}(v) = \won$ holds to a winning strategy of player \emph{even} in the antagonistic update game.

\begin{lemma}
If $\iota_{\min}(v) = \won$, then  player \emph{even} has a strategy to win the antagonistic update game when starting from $v$.
\end{lemma}

\begin{proof}
We can construct the strategy in the following way:
starting in state $v_n =v$, where $n$ is the length of the play we will create, player \emph{even} selects for a state $v_i \in V_e$ with $i>0$ a successor $v_{i-1}$ such that $\iota_i(v_i) \sqsubseteq \au(\iota_{i-1}(v_{i-1}),v_i)$.
Note that such a successor must always exist.
Note also that, if  $v_i \in V_o$ with $i>0$, then $\iota_i(v_i) \sqsubseteq \au(\iota_{i-1}(v_{i-1}),v_i)$ holds for all successors $v_{i-1}$ of $v_i$ by definition.

Assume that player \emph{even} selects a successor from her vertices as described above, and $v_n, v_{n-1}, \ldots, v_0$ is a play created this way.
Let $\mathbf b_0 = \_,\ldots,\_$ be the minimal element of $\mathbb W$, and $\mathbf b_{i+1} = \au(\mathbf b_i,v_{i+1})$.
Then we show by induction that $\mathbf b_i \sqsupseteq \iota_i(v_i)$.

For the induction basis, we have $\mathbf b_0 = \iota_0(v_0)$ by definition.
For the induction step, we have $\iota_{i+1}(v_{i+1}) \sqsubseteq \au(\iota_{i}(v_i),v_{i+1}) \sqsubseteq^{IH} \au(\mathbf b_i,v_{i+1}) = \mathbf b_{i+1}$.

Thus, we get $\mathbf b_n \sqsupseteq \iota_n(v_n) = \won$, and player \emph{even} wins the antagonistic update game.
\end{proof}

At the same time, player \emph{even} cannot win from any vertex $v$ with $\iota_{\min}(v) \neq \won$, and $\iota_{\min}$ provides a witness strategy for player \emph{odd} for this.

\begin{lemma}
Player \emph{even} cannot win from any vertex $v$ with $\iota_{\min}(v) \neq \won$, and $\iota_{\min}$ provides a witness strategy for player \emph{odd}.
\end{lemma}

\begin{proof}
We recall that the construction of $\iota_{\min}$ by Lemma \ref{lem:construct} provides
\begin{itemize}
 \item $\iota_{\min}(v) \sqsubseteq \max_\sqsubseteq\{\au(\iota_{\min}(s),v) \mid (v,s) \in E\}$ for $v \in V_e$, and
 \item $\iota_{\min}(v) \sqsubseteq \min_\sqsubseteq\{\au(\iota_{\min}(s),v) \mid (v,s) \in E\}$ for $v \in V_o$.
\end{itemize}
The latter provides the existence of some particular successor $s$ of $v$ with $\iota_{\min}(v) \sqsubseteq \au(\iota_{\min}(s),v)$.
The witness strategy of player \emph{odd} is to always choose such a vertex.

Let $\rho = v_n,v_{n-1},v_{n-2},\ldots,v_1$ be a sequence obtained by any strategy of player \emph{even} from a starting vertex $v_n$ with $\iota_{\min}(v_n) \neq \won$, such that player \emph{even} chooses to evaluate the backward antagonistic update game after $\rho$, and $\rho,v_0$ an extension in line with the strategy of player odd.

We first observe that $\iota_{\min}(v_{i+1}) \sqsubseteq \au(\iota_{\min}(v_i),v_{i+1})$ holds for all $i<n$, either by the choice of the successor of $v_{i+1}$ of player \emph{odd} if $v_{i+1}\in V_o$, or by 
$\iota_{\min}(v_{i+1}) \sqsubseteq \max_\sqsubseteq\{\au(\iota_{\min}(s),v_{i+1}) \mid (v_{i+1},s) \in E\} \sqsubseteq \au(\iota_{\min}(v_i),v_{i+1})$ if $v_{i+1}\in V_e$.
With $\iota_{\min}(v_n) \neq \won$, this provides $\iota_{\min}(v_i) \neq \won$ for all $i \leq n$.

Let $\mathbf b_0 = \_,\ldots,\_$ be the minimal element of $\mathbb W$, and $\mathbf b_{i+1} = \au(\mathbf b_i,v_{i+1})$.
Then $\mathbf b_0 \sqsubseteq \iota_{\min}(v_0)$, and the monotonicity of $\au$ in the first element inductively provides $\mathbf b_i \sqsubseteq \iota_{\min}(v_i)$ for all $i \leq n$.
Thus $\mathbf b_n \neq \won$, and player \emph{even} loses the update game.
\end{proof}

\section{Complexity}
We use natural representation for the set of colours as integers written in binary, encoding the $\_$ as $0$.
The first observation is that the number of individual lift operations is, for each vertex, limited to $|\mathbb W|$.

\begin{lemma}
\label{lem:number.of.operations}
For each vertex the number of lift operations is restricted to  $|\mathbb W|$.
The overall number of lift operations is restricted to $|V|\cdot |\mathbb W|$.
The number of lift operations an edge (or: source or target vertex of an edge, respectively) is involved in is restricted to   $|\mathbb W|$.
Summing up over all edges and over the number of lift operations their target \emph{or} source vertex is involved in amounts to $O(|E| \cdot |\mathbb W|)$.
\end{lemma}

A simple implementation can track, for each vertex, the information which position in the \is\ is the next one that would need to be updated to trigger a lift along an edge, and, using a binary representation in line with $\succcurlyeq$, which bit in the representation of this position has to change to consider triggering an update. (Intuitively the most significant bit that separates the current value from the next value that would trigger an update.)

Obviously, the most expensive path to $\iota_{\min}$ is for each position to go through all values of $|\mathbb W|$ in this case.
But in this case, tracking the information mentioned in the previous section reduces the average cost of an update to $O(1)$.
The information that we store for this is, for each vertex, the current \is\ that represents its current value before and after executing the antagonistic update, and the next value that would lead to a lift operation on the antagonistic value.

For each incoming edge, the position and bit that need to be increased to trigger the next lift operation for this vertex are also stored.

\begin{example}
We look at a vertex $v$ with one outgoing edge to its successor vertex $s$.
We have $7$ different colours, $2$ through $8$. Vertex $v$ has colour $2$.

We use a representation that follows the $\succcurlyeq$ order and thus maps $0$ to $0$, $7$ to $1$, $5$ to $2$, $3$ to $3$, $2$ to $4$, $4$ to $5$, $6$ to $6$, $8$ to $7$.

Assume that $s$ has currently a witness $\mathbf b = b_2,b_1,b_0 = 6,0,2$ attached to it, represented as $\widetilde{\mathbf b}=\widetilde{b}_2,\widetilde{b}_1,\widetilde{b}_0 = 6,0,4$.

To obtain a witness for $v$, we calculate $\mathbf c = \au(\mathbf b, v) = 6,5,2$, which is represented as $\widetilde{\mathbf c}=\widetilde{c}_2,\widetilde{c}_1,\widetilde{c}_0 = 6,2,4$.
The next higher value $\mathbf a \sqsupset \mathbf b$ such that $\au(\mathbf a, v) \sqsupset \au(\mathbf b, v)$ is $\widetilde{\mathbf a}=\widetilde{a}_2,\widetilde{a}_1,\widetilde{a}_0 = 6,2,4$.

The lowest position $i$ with $\widetilde{a}_i > \widetilde{b}_i$ is for position $i=1$, and the difference occurs in the middle bit ($\widetilde{a}_1 = 2 = 010_2$ and $\widetilde{b}_1 = 0 = 000_2$).

For the edge from $v$ to $s$, we can store after the update that we only need to consider an update from $s$ if it increases at least the position $b_1$ of the witness for $s$. If $b_1$ is changed, we only have to consider the change if the update is at least to the value represented as 2 ($\widetilde{b}_1' \geq 2$), and thus $b_1' \succcurlyeq 5$.
For all smaller updates of the witness of $s$, no update of the witness of $v$ needs to be considered.
\end{example}

\begin{theorem}
\label{theo:cost}
For a parity game with $n$ vertices and $m$ edges, the algorithm can be implemented to run in $O(m \cdot |\mathbb W|)$ time and $O(n \cdot \log |\mathbb W| + m \log \log |\mathbb W|)$ space.
\end{theorem}

Note that the $\log \log |\mathbb W|$ information per edge is only required to allow for a discounted update cost of $O(1)$. It can be traded for a $\log |\mathbb W|$ increase in the running time.
This leaves the estimation of $|\mathbb W|$.

To improve the complexity especially in the relevant lower range of colours, we first look into reducing the size of $\mathbb W$, and then look into keeping the discounted update complexity low.
We make three observations that can be used to reduce the size of $\mathbb W$; they can be integrated in the overall proof, starting with the raw and basic update steps.

The first observation is that, if the highest colour is the odd colour $o_{\max}$, then we do not need to represent this colour:
if $\phi(v)= o_{\max}$ and $\mathbf b \neq \won$, then $\up(\mathbf b,v)$ contains only $\_$ and $o_{\max}$ entries.
Moreover, $\_$ and $o_{\max}$ entries behave in exactly the same way.
This is not surprising: $o_{\max}$ is the most powerful colour, and a state with colour $o_{\max}$ cannot occur on a winning cycle.

The second observation is that, if the lowest colour is the odd colour $o_{\min}$, then we can ignore it during all update steps without violating the correctness arguments.
(In fact, this colour cannot occur at all when using the update rules suggested
in Calude et al.~\cite{quasipolynomial}.)

Finally, we observe that, for the least relevant entry $b_0$ of an \is\ $\mathbf b$, it does not matter if this entry contains $\_$ or an odd value.
We can therefore simply not use odd values at this position.
(Using the third observation has no impact on the complexity of the problem, but
still approximately halves the size of $\mathbb W$, and is therefore useful in practice.)

We call the number of different colours, not counting the maximal and minimal
colour if they are odd, the number $r$ of \emph{relevant colours}.

\begin{lemma}
For a parity game with $r$ relevant colours and $e$ vertices with even colour, and thus with length $l=\lceil \log_2(e+1)\rceil$ of the \is es,  $|\mathbb W| \leq 1 + \sum\limits_{i=0}^l \Big(\begin{array}{c} l \\ i \end{array}\Big) \cdot \Big(\begin{array}{c}i + r - 1  \\ r-1 \end{array}\Big)$.
\end{lemma}

\begin{proof}
The $1$ refers to the dedicated value $\won$. For the other \is es, the values can be obtained
by considering the number $i$ of integer entries.
For $i$ integer entries, there are $\Big(\begin{array}{c} l \\ i \end{array}\Big)$ different positions in the \is es that could hold these $i$ integer values.
Fixing these positions, there are $\Big(\begin{array}{c}i + r - 1 \\ r-1 \end{array}\Big)$ ways to assign non-increasing values from the range of relevant colours.
(E.g.\ these can be represented by a sequence of $i$ white balls and $r-1$ black balls.
The number of white balls prior to the first black ball is the number of positions assigned the highest relevant colour, the number of white balls between the first and second black ball is the number of positions assigned the next lower colour, etc.)
\end{proof}

This allows for two easy estimations of the size of $|\mathbb W|$: 
If the number $c$ of colours is small (especially if $c$ is constant), then we can use the coarse estimate $|\mathbb W| \in O\Big( e \cdot  \Big(\begin{array}{c}l + r - 1  \\ l \end{array}\Big) \Big)$.

In particular, we get the following complexity for a constant number of colours.

\begin{theorem}
\label{theo:simplified.cost}
A parity game with $r$ relevant colours, $n$ vertices, $m$ edges, and $e$ vertices with even colour can be solved in time $O\big(e\cdot m \cdot (\log (e) + r)^{r-1} / (r-1)!\big)$ and space $O\big(n \cdot \log (e) \cdot \log (r) + m \cdot \log( \log (e) \cdot \log (r))\big)$.
\end{theorem}

We use that the length $l = \lceil \log_2 (e+1) \rceil$ of the \is es is logarithmic in $e$.

This also provides us with a strong fixed parameter tractability result:
when we fix the number of colours to some constant $c$, we maintain a quasi bi-linear complexity in the number of edges and the number of vertices.
If we fix, e.g., a monotonously growing quasi constant function $\mathsf{qc}$ (like the inverse Ackermann function), then Theorem \ref{theo:simplified.cost} shows that, as soon $\mathsf{qc}(n) \geq c$, and thus almost everywhere and in particular in the limit, have $(l+r)^{r-1} / (r-1)! \leq  (\log_2 n)^{\mathsf{qc}(n)}$, or $(l+r)^{r-1} / (r-1)! \leq \mathsf{qc}(n)^{\log_2(\log_2(n))}$ if $\log_2(\mathsf{qc}(n) \geq c)$.

\begin{corollary}
Parity games are fixed parameter tractable, using the number of colours as their
parameter, with complexity $O\big(m \cdot n \cdot \mathsf{qc}(n)^{\log \log
n}\big)$ for an arbitrary quasi constant $\mathsf{qc}$, where $m$ is the number of edges and $n$ is the number of states.
\end{corollary}

For a ``high'' number of colours, we can improve the estimation:
if $r \geq l^2$, then the case $i = l$ dominates the overall cost, such that $|\mathbb W| \in O\Big(\big(\begin{array}{c}l + r - 1  \\ l \end{array}\big)\Big)$.

\begin{theorem}
For a parity game with $r$ relevant colours, $m$ edges, and $e$ vertices with even colour, and thus length $l = \lceil \log_2(e+1) \rceil$ of the \is es, and $h = \big\lceil 1 + \frac{r-1}{l}\big\rceil$, one can solve the parity game in time
$O(m \cdot h \cdot e^{1+c_{1.45}+\log_2(h)})$, and in time $O(m \cdot h \cdot e^{c_{1.45}+\log_2(h)})$ if $r > l^2$.
\end{theorem}

We use the constant $c_{1.45} = \lim_{h\rightarrow \infty} \log_2(1+1/h) 
\cdot h = \log_2 \mathsf e < 1.45$, 
where $\mathsf e \approx 2.718$ is the Euler number; using that
$(1+1/h)^h < \mathsf e$ and thus $\log_2  (1+1/h) \cdot h < c_{1.45}$ holds for all $h \in \mathbb N$.

\begin{proof}
To estimate $\mathbb W$, we again start with analysing the size of $\Big(\begin{array}{c}l + r - 1  \\ l \end{array}\Big)$.

We note that $l + r - 1 \leq h\cdot l$, such that we can estimate this value by drawing $l$ out of $h\cdot l$.

The number of all ways to choose
$l=\lceil\log(e+1)\rceil$ out of $h\cdot l$ numbers can, by the
Wikipedia page on binomial coefficients and the inequality using the entropy in there (also can be found in \cite{robert1990ash}), be bounded by
\begin{eqnarray*}
    && 2^{(\log_2(e)+1) \cdot h \cdot
       ((1/h) \cdot \log_2(h) + ((h-1)/h) \cdot
       \log_2(h/(h-1)))} \\
    & = & 2^{(\log_2(e)+1) \cdot (\log_2(h)+\log_2(1+1/(h-1))\cdot (h-1))} \\
    & = & (2e)^{\log_2(h)+(\log_2(1+1/(h-1))) \cdot (h-1))} \\
    & \leq & (2e)^{c_{1.45}+\log_2(h)} \in O\big(h \cdot e^{c_{1.45}+\log_2(h)}\big).
\end{eqnarray*}

The estimation uses that $\log(1+1/(h-1)) \cdot (h-1) < c_{1.45}$ holds for all $h \in \mathbb N$.

Theorem \ref{theo:cost} now provides $O(m \cdot h \cdot e^{1+c_{1.45}+\log_2(h)})$ time bound.
If the number of colours is high ($r > l^2$), then we observe that
$|\mathbb W| \leq 1 + \sum_{i=0}^l \Big(\begin{array}{c} l  \\ i \end{array}\Big)\cdot \Big(\begin{array}{c}i + r - 1  \\ i \end{array}\Big) \in O\Big(\big(\begin{array}{c}l + r - 1  \\ l \end{array}\big)\Big)$ holds,
as the sum is dominated by $\Big(\begin{array}{c} l  \\ l \end{array}\Big)\cdot \Big(\begin{array}{c}l + r - 1  \\ l \end{array}\Big)$.
This allows for the second estimate $O(m \cdot h \cdot e^{c_{1.45}+\log_2(h)})$ of the running time when $r > l^2$ holds.
\end{proof}

This allows for identifying a class of parity games that can be solved in polynomial time.

\begin{corollary}
Parity games where the number $c$ of colours is logarithmically bounded by the number $e$ of vertices with even colour ($c \in O(\log e)$) can be solved in polynomial time.
\end{corollary}

\section{Lower Bounds}
In this section, we introduce a family of examples, on which the the basic update game from \cite{quasipolynomial} is slow.
(Recall that these original rules restrict the use of Lemma \ref{lem:up.overflow} to even colours.
Adjusting the example is not hard, but effectively disallows to make effective use of $b_0$.)

The example is a single player game, which is drawn best
as a ring. In this example, the losing player, player \emph{odd}, can draw out his loss. The vertices
of the game have name and colour $1,\ldots,2n$.
They are all owned by player \emph{odd}.
There is always an edge to the next vertex (in the modulo ring).
Additionally, there is an edge back to $1$ from all vertices with even name (and colour).

Obviously, all runs are winning for player \emph{even}. 
We show how player \emph{odd} can, when starting in vertex $1$, produce a play, such that forward updates produce all \is es that use only $\_$ and even numbers.

We first observe that every value $2i-1$ is overwritten after the next move in a play by $2i$ in a witness $\mathbf b$.

The strategy of player \emph{odd} to create a long path is simple.
We consider three cases.

If, in the current witness $\mathbf b = b_k,\ldots,b_0$, we have $b_0=\_$ 
and the token is at a position $2i$,
then moving to $1$, and thus next to $2$, results in the next larger witness without odd entries than $\bfb$.

If  $b_0\neq \_$, then we have that $b_0 = 2i$, and $\bfb$ has no smaller entries than $2i$.
If all of these entries are consecutively on the right of $\bfb$, then we obtain the next larger witness without odd entries than $\bfb$ by going through $2i+1$ to $2i+2$. Player odd therefore chooses to continue by moving the token to vertex $2i + 1$ in this case.

Otherwise, there is a rightmost $b_j = \_$, such that right of it are only entries $2i$ (for all $h<j$, $b_h = 2i$), and there is also a $2i$ value to the left (for some $h>j$, $b_h = 2i$).
Then the next larger witness without odd entries than $\bfb$ is obtained by replacing $b_j$ by $2$ and all entries to its right by $\_$. This can be obtained by going to vertex 1 and, subsequently, to vertex 2. Player odd therefore chooses to continue by moving the token to vertex $1$ in this case.

\begin{figure}
\begin{center}
\begin{tikzpicture}[node distance = 1.5cm]
\tikzstyle{vertex} = [minimum size=0.7cm, inner sep=-0.2cm,draw, circle]

\node [vertex] (1) {$1$};
\node [vertex,below right of=1] (2) {$2$};
\node [vertex,below left of=2] (3) {$3$};
\node [vertex,above left of=3] (4) {$4$};

\path[->] 
    (1) edge [bend right=45] (2)
    (2) edge (3)
    (3) edge (4)

    (4) edge (1)
    (2) edge [bend right=45] (1)
    ;
\end{tikzpicture}
\end{center}
\caption{The lower bound example for $n = 2$.}
\end{figure}
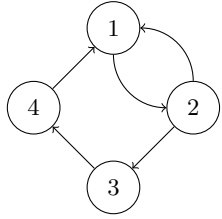

\section{Implementation}
We implemented our algorithm in {\sc C++} and tested its performance 
on Mac OS X with 1.7 GHz Intel Core i5 CPU and 4 GB of RAM.
We then compared it with the small progress measure algorithm \cite{Jurdzinski/00/ParityGames}, Zielonka's recursive algorithm \cite{Zielonka/98/Parity}, the classic strategy improvement algorithm \cite{Voge+Jurdzinski/00/simprove} as implemented in the {\sc PGSolver version 4.0} \cite{friedmann2014pgsolver,pgsolver17},
and the implementation \cite{totzke17} of an alternative recently developed succinct progress measure algorithm  from \cite{JL17}. 
We tested their performance, with timeout set to two minutes, on around 250 different parity games of various sizes generated using {\sc PGSolver}.
These examples include the following classes.
\begin{itemize}
\item Friedmann's trap examples~\cite{Friedmann/11/lower}, which show exponential lower bound for the classic strategy improvement algorithm;
\item random parity games of sizes, $s$, ranging from 100 to 10000 that were generated using {\sc PGSolver}'s command
\texttt{steadygame $s$ 1 6 1 6} (for each $s$ we generated ten instances);
\item recursive ladder construction \cite{friedmann2011recursive} generated using 
{\sc PGSolver}'s command \texttt{recursiveladder}.
\end{itemize}

{\sc PGSolver} implements several optimisation steps before the algorithm of choice is invoked.
These include SCC decomposition, detection of special cases, priority compression, and priority propagation
as described in \cite{friedmann2014pgsolver}.
To illustrate this, the small progress measures algorithm in {\sc PGSolver} was able to solve all 
Friedmann's trap examples in 0.01 seconds when using these optimisations.
However, without these optimisations, it failed to terminate within the set timeout of two minutes. 
As our aim was to compare different algorithms and not the heuristics or preprocessing steps involved, 
we invoked {\sc PGSolver} with options ``\texttt{-dgo -dsd -dlo -dsg}'' to switch off some of these optimisation steps.
We believe this gives a better and fairer picture of the relative performance of these algorithms.
Some of these optimisations are embedded in the algorithms themselves and cannot be switched off. 
For example, the small progress measure algorithm implemented in {\sc PGSolver} starts off with the computation of maximal values that may ever need to be considered \cite{friedmann2014pgsolver}. 
In future, we plan to include these optimisation preprocessing techniques into our tool as well.


\begin{table*}
\caption{Running times (in seconds) of the four algorithms tested: quasi-polynomial time algorithm presented in this paper (QPT), small progress measure (SPM), Zielonka's recursive algorithm (REC), the classic strategy improvement (CSI), and the implementation \cite{totzke17} of the quasi-polynomial time algorithm (JL'17) from \cite{JL17}.
Entry ``--'' means that the algorithm did not terminate within the set timeout of two minutes.
For the \texttt{steadygame} examples we state the minimum and the maximum measured execution time for the ten examples generated for each size.
}
\hspace*{-1em}
\begin{tabular}{ccclllll}
\toprule
Example Class & Nodes & Colours & \multicolumn{1}{c}{QPT} & \multicolumn{1}{c}{SPM} & \multicolumn{1}{c}{REC} & \multicolumn{1}{c}{CSI} & \multicolumn{1}{c}{JL'17 \cite{JL17,totzke17}} \\
\midrule
\multirow{2}{*}{\texttt{steadygame}} & \multirow{2}{*}{100} & \multirow{2}{*}{100} & min:\, 0.01 & min:\, 0.01 & min:\, 0.01 & min:\, 0.01 & min:\, 10.16 \\
 &  & & max: 0.02 & max: 0.02 & max: 0.01 & max: 0.01 & max: -- \\
\multirow{2}{*}{\texttt{steadygame}} & \multirow{2}{*}{200} & \multirow{2}{*}{200} & min:\, 0.01 & min:\, 0.01 & min:\, 0.01 & min:\, 0.01 & min:\, -- \\
 &  & & max: 0.09 & max: 0.06 & max: 0.01 & max: 0.03 & max: -- \\
\multirow{2}{*}{\texttt{steadygame}} & \multirow{2}{*}{1000} & \multirow{2}{*}{1000} & min:\, 0.09 & min:\, 1.55 & min:\, 0.01 & min:\, 0.14 & min:\, -- \\
 &  & & max: 1.51 & max: 1.67 & max: 0.04 & max: 0.23 & max: -- \\
\multirow{2}{*}{\texttt{steadygame}} & \multirow{2}{*}{5000} & \multirow{2}{*}{5000} & min:\, 1.51 & min:\, 41.49 & min:\, 0.23 & min:\, 1.56 & min:\, -- \\
 &  & & max: 102 & max: -- & max: 0.44 & max: 4.12 & max: -- \\
\multirow{2}{*}{\texttt{steadygame}} & \multirow{2}{*}{10000} & \multirow{2}{*}{10000} & min:\, 5.1 & min:\, -- & min:\, 0.68 & min:\, 3.07 & min:\, -- \\
 &  & & max: -- & max: -- & max: 1.89 & max: 8.25 & max: -- \\
\texttt{Friedmann's trap} & 77 & 66 & 0.01 & -- & 0.01 & 0.26 & -- \\
\texttt{Friedmann's trap} & 230 & 120 & 0.01 & -- & 0.01 & 22.72 & -- \\
\texttt{Friedmann's trap} & 377 & 156 & 0.01 & -- & 0.01 & -- & -- \\
\texttt{recursive ladder}     & 250 & 152 & 0.01 & -- & -- & 0.01 & 0.66 \\
\texttt{recursive ladder}     & 1000 & 752 & 0.02 & -- & -- & 0.01 & -- \\
\texttt{recursive ladder}     & 25000 & 15002 & 0.45 & -- & -- & 0.56 & -- \\
\bottomrule
\end{tabular}
\label{tab:results}
\end{table*}

The more interesting results of our tests are presented in Table \ref{tab:results}. 
As expected, our algorithm is outperformed by strategy improvement and recursive algorithm on randomly generated examples.
Our algorithm is very fast on Friedmann's trap examples, because player odd wins from all nodes
and a fixed point is reached very quickly using a small number of entries in the witnesses.
Finally, we tested the algorithms on the recursive ladder construction, which is a class of examples for which the recursive algorithm runs in exponential time.
As expected, the small progress measure and the recursive algorithm fail to terminate for 
examples as small as 250 nodes. Our algorithm as well as the classic strategy improvement
solved these instances very quickly. 
Interestingly, the worst performing algorithm is \cite{JL17}, which currently has the best theoretical upper bound on its running time.
The most likely reason for this is that their single step of the value iteration is a lot more complicated than ours. As a result, even if less such steps are required to reach a fixed point, the algorithm performs badly as each step is a lot slower.
In conclusion, our algorithm complements quite well the existing well-established algorithms for parity games and can be faster than any of them depending on the class of examples being considered.

The implementation of our algorithm along with all the examples that we used in this comparison are available at \url{https://cgi.csc.liv.ac.uk/~dominik/parity/}.

\section*{Acknowledgements}
We thank Ding Xiang Fei for pointing out an error in a previous
version of this manuscript.
Sanjay Jain was supported in part by NUS grant C252-000-087-001.
Further, Sanjay Jain and Frank Stephan were supported in
part by the Singapore Ministry of Education Academic Research Fund Tier 2
grant MOE2016-T2-1-019 / R146-000-234-112.
Sven Schewe and Dominik Wojtczak were supported in part by EPSRC grant EP/M027287/1.
Further, John Fearnley, Sven Schewe and Dominik Wojtczak were supported in part by EPSRC grant EP/P020909/1.

\bibliographystyle{abbrv}
\bibliography{bib}

\begin{thebibliography}{10}

\bibitem{Alur+Henziger+Kupferman/02/ATL}
R.~Alur, T.~A. Henzinger, and O.~Kupferman.
\newblock Alternating-time temporal logic.
\newblock {\em Journal of the ACM}, 49(5):672--713, 2002.

\bibitem{Berwanger+all/06/ParityDAG}
D.~Berwanger, A.~Dawar, P.~Hunter, and S.~Kreutzer.
\newblock Dag-width and parity games.
\newblock In {\em Proc.\ of STACS}, pages 524--436. Springer-Verlag, 2006.

\bibitem{BjorklundVorobyov/07/subexp}
H.~Bj{\"o}rklund and S.~Vorobyov.
\newblock A combinatorial strongly subexponential strategy improvement
  algorithm for mean payoff games.
\newblock {\em Discrete Appl. Math.}, 155(2):210--229, 2007.

\bibitem{Browne-all/97/fixedpoint}
A.~Browne, E.~M. Clarke, S.~Jha, D.~E. Long, and W.~Marrero.
\newblock An improved algorithm for the evaluation of fixpoint expressions.
\newblock {\em TCS}, 178(1--2):237--255, 1997.

\bibitem{quasipolynomial}
C.~S. Calude, S.~Jain, B.~Khoussainov, W.~Li, and F.~Stephan.
\newblock Deciding parity games in quasipolynomial time.
\newblock In {\em Proc.\ of STOC 2017}, page (to appear). ACM Press, 2017.

\bibitem{DBLP:conf/lics/ChatterjeeHL15}
K.~Chatterjee, M.~Henzinger, and V.~Loitzenbauer.
\newblock Improved algorithms for one-pair and k-pair streett objectives.
\newblock In {\em Proc. of LICS}, pages 269--280. {IEEE} Computer Society,
  2015.

\bibitem{deAlfaro+Henziger+Majumdar/01/Control}
L.~de~Alfaro, T.~A. Henzinger, and R.~Majumdar.
\newblock From verification to control: Dynamic programs for omega-regular
  objectives.
\newblock In {\em Proc.\ of LICS}, pages 279--290, 2001.

\bibitem{Emerson+Jutla/91/Memoryless}
E.~A. Emerson and C.~S. Jutla.
\newblock Tree automata, $\mu$-calculus and determinacy.
\newblock In {\em Proc.\ of FOCS}, pages 368--377. IEEE Computer Society Press,
  October 1991.

\bibitem{Emerson+all/93/mu}
E.~A. Emerson, C.~S. Jutla, and A.~P. Sistla.
\newblock On model-checking for fragments of $\mu$-calculus.
\newblock In {\em Proc.\ of CAV}, pages 385--396, 1993.

\bibitem{Emerson+Lei/86/Parity}
E.~A. Emerson and C.~Lei.
\newblock Efcient model checking in fragments of the propositional
  $\mu$-calculus.
\newblock In {\em Proc.\ of LICS}, pages 267--278. IEEE Computer Society Press,
  1986.

\bibitem{Fearnley/10/snare}
J.~Fearnley.
\newblock Non-oblivious strategy improvement.
\newblock In {\em Proc.\ of LPAR}, pages 212--230, 2010.

\bibitem{Friedmann/11/lower}
O.~Friedmann.
\newblock An exponential lower bound for the latest deterministic strategy
  iteration algorithms.
\newblock {\em LMCS}, 7(3), 2011.

\bibitem{friedmann2011recursive}
O.~Friedmann.
\newblock Recursive algorithm for parity games requires exponential time.
\newblock {\em RAIRO-Theoretical Informatics and Applications}, 45(4):449--457,
  2011.

\bibitem{FriedmannL09}
O.~Friedmann and M.~Lange.
\newblock Solving parity games in practice.
\newblock In {\em Proc.\ of ATVA}, pages 182--196, 2009.

\bibitem{friedmann2014pgsolver}
O.~Friedmann and M.~Lange.
\newblock The pgsolver collection of parity game solvers.
\newblock {\em University of Munich}, 2014.
\newblock http://www.win.tue.nl/~timw/downloads/amc2014/pgsolver.pdf.

\bibitem{pgsolver17}
O.~Friedmann and M.~Lange.
\newblock Pgsolver version 4.0, 2017.

\bibitem{DBLP:conf/cav/HahnSTZ16}
E.~M. Hahn, S.~Schewe, A.~Turrini, and L.~Zhang.
\newblock A simple algorithm for solving qualitative probabilistic parity
  games.
\newblock In {\em Proc.\ of CAV}, volume 9780 of {\em LNCS}, pages 291--311,
  2016.

\bibitem{Jurdzinski/98/UP}
M.~Jurdzi{\'n}ski.
\newblock Deciding the winner in parity games is in {UP}~$\cap$~co-{UP}.
\newblock {\em Information Processing Letters}, 68(3):119--124, November 1998.

\bibitem{Jurdzinski/00/ParityGames}
M.~Jurdzi\'nski.
\newblock Small progress measures for solving parity games.
\newblock In {\em Proc.\ of STACS}, pages 290--301. Springer-Verlag, 2000.

\bibitem{JL17}
M.~Jurdzi{\'n}ski and R.~Lazi{\'c}.
\newblock Succinct progress measures for solving parity games.
\newblock In {\em Proc.\ of LICS 2017}, page (to appear), 2017.

\bibitem{Jurdzinski/06/subex}
M.~Jurdzi{\'n}ski, M.~Paterson, and U.~Zwick.
\newblock A deterministic subexponential algorithm for solving parity games.
\newblock In {\em Proc.\ of SODA}, pages 117--123. ACM/SIAM, 2006.

\bibitem{Kozen/83/mu}
D.~Kozen.
\newblock Results on the propositional $\mu$-calculus.
\newblock {\em TCS}, 27:333--354, 1983.

\bibitem{Lange/05/ParitySAT}
M.~Lange.
\newblock Solving parity games by a reduction to {SAT}.
\newblock In {\em Proc.\ of Int.\ Workshop on Games in Design and
  Verification}, 2005.

\bibitem{Ludwig/95/random}
W.~Ludwig.
\newblock A subexponential randomized algorithm for the simple stochastic game
  problem.
\newblock {\em Inf. Comput.}, 117(1):151--155, 1995.

\bibitem{McNaughton/93/Games}
R.~McNaughton.
\newblock Infinite games played on finite graphs.
\newblock {\em Ann. Pure Appl. Logic}, 65(2):149--184, 1993.

\bibitem{Obdrzalek/03/TreeWidth}
J.~Obdr\v{z}\'alek.
\newblock Fast mu-calculus model checking when tree-width is bounded.
\newblock In {\em Proc.\ of CAV}, pages 80--92. Springer-Verlag, 2003.

\bibitem{Piterman/06/Parity}
N.~Piterman.
\newblock From nondeterministic {B}\"uchi and {S}treett automata to
  deterministic parity automata.
\newblock In {\em Proc.\ of LICS}, pages 255--264. IEEE Computer Society, 2006.

\bibitem{Puri/95/simprove}
A.~Puri.
\newblock {\em Theory of hybrid systems and discrete event systems}.
\newblock PhD thesis, Computer Science Department, University of California,
  Berkeley, 1995.

\bibitem{robert1990ash}
A.~B. Robert.
\newblock Information theory, 1990.

\bibitem{Schewe/08/improvement}
S.~Schewe.
\newblock An optimal strategy improvement algorithm for solving parity and
  payoff games.
\newblock In {\em Proc.\ of CSL 2008}, pages 368--383. Springer-Verlag, 2008.

\bibitem{Schewe/17/parity}
S.~Schewe.
\newblock Solving parity games in big steps.
\newblock volume~84, pages 243--262, 2017.

\bibitem{Schewe+Finkbeiner/06/ATM}
S.~Schewe and B.~Finkbeiner.
\newblock The alternating-time $\mu$-calculus and automata over concurrent game
  structures.
\newblock In {\em Proc.\ of CSL}, pages 591--605. Springer-Verlag, 2006.

\bibitem{Schewe+Finkbeiner/06/Asynchronous}
S.~Schewe and B.~Finkbeiner.
\newblock Synthesis of asynchronous systems.
\newblock In {\em Proc.\ of LOPSTR}, pages 127--142. Springer-Verlag, 2006.

\bibitem{DBLP:conf/icalp/ScheweTV15}
S.~Schewe, A.~Trivedi, and T.~Varghese.
\newblock Symmetric strategy improvement.
\newblock In {\em Proc. of ICALP}, volume 9135 of {\em LNCS}, pages 388--400,
  2015.

\bibitem{totzke17}
P.~Totzke.
\newblock Implementation of the succinct progress measures algorithm from
  \cite{JL17}, 2017.

\bibitem{Vardi/98/2WayAutomata}
M.~Y. Vardi.
\newblock Reasoning about the past with two-way automata.
\newblock In {\em Proc.\ of ICALP}, pages 628--641. Springer-Verlag, 1998.

\bibitem{Voge+Jurdzinski/00/simprove}
J.~V{\"{o}}ge and M.~Jurdzi{\'{n}}ski.
\newblock A discrete strategy improvement algorithm for solving parity games.
\newblock In {\em Proceedings of the CAV}, pages 202--215. Springer-Verlag,
  July 2000.

\bibitem{Wilke/01/Alternating}
T.~Wilke.
\newblock Alternating tree automata, parity games, and modal $\mu$-calculus.
\newblock {\em Bull.\ Soc.\ Math.\ Belg.}, 8(2), May 2001.

\bibitem{Zielonka/98/Parity}
W.~Zielonka.
\newblock Infinite games on finitely coloured graphs with applications to
  automata on infinite trees.
\newblock {\em Theor. Comput. Sci.}, 200(1-2):135--183, 1998.

\end{thebibliography}

\appendix
\onecolumn
\begin{figure}
\includegraphics[width=\textwidth]{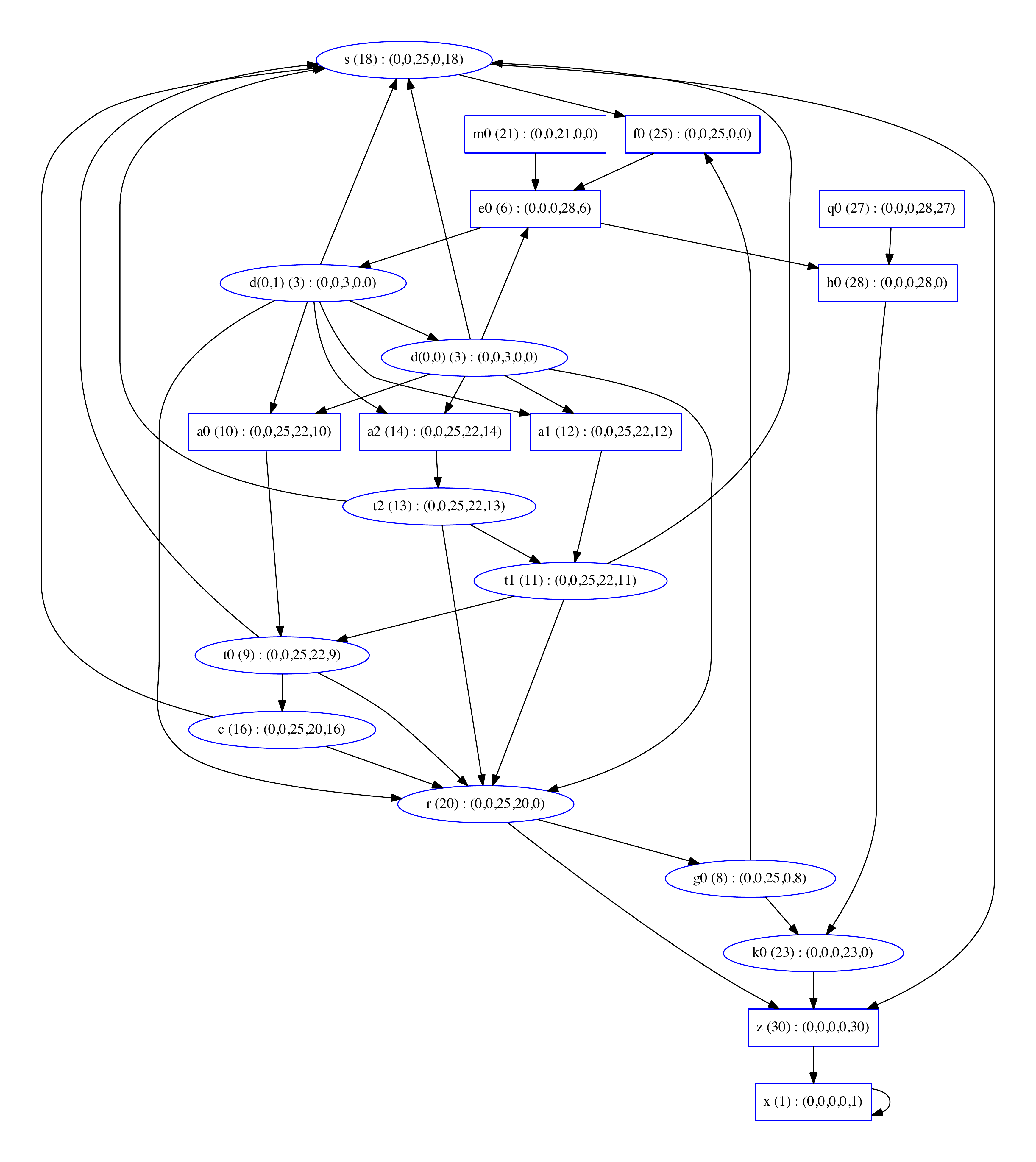}
\caption{The fixed-point reached when using the QPT algorithm to solve the Friedmann's trap example with 20 nodes.
Square nodes belong to player odd and circle nodes to player even. The label of a node consists of
its name, followed by its colour (in parentheses), and after a colon its witness for $\iota_{\min}$.}
\label{fig:trap}
\end{figure}

\end{document}